\documentclass[11pt,transaction,twocolumn]{IEEEtran}
\usepackage{graphics}
\usepackage{graphicx}
\usepackage{amsmath,amsfonts,amssymb,varioref,mathrsfs,verbatim,amsthm}
\usepackage{latexsym}
\usepackage{epsfig}
\usepackage[nospace,noadjust]{cite}
\usepackage{latexsym}
\usepackage{epsfig}
\usepackage{epstopdf}
\usepackage[nospace,noadjust]{cite}
\usepackage{array}
\usepackage[english]{babel}
\usepackage{color}
\usepackage{babel}
\usepackage{algorithm,algpseudocode}
\usepackage{xcolor,colortbl}
\newtheorem{theorem}{{\bf Theorem}}
\newtheorem{corollary}{{\bf Corollary}}
\newtheorem{proposition}{\noindent {\bf Proposition}}
\newtheorem{lemma}{\noindent {\bf Lemma}}
\newtheorem{remark}{\noindent {\bf Remark}}
\newtheorem{mydef}{\noindent {\bf Definition}}

\newcommand\myeq{\mathrel{\overset{\makebox[0pt]{\mbox{\normalfont\tiny\sffamily def}}}{=}}}
\include{new_commands}
\newcommand{\diag}{\mathop{\rm diag}}

\makeatletter
\let\l@ENGLISH\l@english
\makeatother

\begin{document}%
\title{Analytical Derivation of the Inverse Moments of  One-sided Correlated Gram Matrices with Applications}
\author{Khalil Elkhalil,~\IEEEmembership{Student Member,~IEEE,} Abla Kammoun,~\IEEEmembership{Member,~IEEE}, Tareq~Y.~Al-Naffouri,~\IEEEmembership{Member,~IEEE}, and Mohamed-Slim Alouini,~\IEEEmembership{Fellow,~IEEE}

\thanks{K. Elkhalil, A. Kammoun, T. Y. Al-Naffouri and M.-S. Alouini are with the Electrical Engineering Program, King Abdullah University of Science and Technology, Thuwal, Saudi Arabia; e-mails: \{khalil.elkhalil, abla.kammoun, tareq.alnaffouri, slim.alouini\}@kaust.edu.sa. Tareq Y. Al-Naffouri is also associated with the Department of Electrical Engineering, King Fahd University of Petroleum
and Minerals, Dhahran 31261, Kingdom of Saudi Arabia.
}

}

\maketitle

\vspace{-15mm}

\begin{abstract}

This paper addresses the development of analytical tools for the computation of the moments of random Gram matrices with one side correlation. Such a question is mainly driven by applications in signal processing and wireless communications wherein such matrices naturally arise. In particular, we derive closed-form expressions for the inverse moments and show that the obtained results  can help approximate  several performance metrics such as the average estimation error corresponding to the Best Linear Unbiased Estimator (BLUE) and the Linear Minimum Mean Square Error LMMSE or also other loss functions  used to measure the accuracy of covariance matrix estimates.

\end{abstract}


\section{Introduction and basic assumptions}

The study of the behavior of random matrices is a key question that appears in many disciplines such as wireless communication, signal processing and  economics, to name a few.  The main motivation behind this question comes from the fundamental role that play random matrices in modeling unknown and unpredictable physical quantities. In many situations, meaningful metrics expressed as scalar functionals of these random matrices naturally arise. The understanding of their behaviour is, however, a difficult task which might be out of reach especially when involved random models are considered.  One approach to tackle this problem is represented by the moment method. It basically resorts to the Taylor expansion of differentiable functionals in order to turn this difficult question into that of computing the moments of random matrices, where the moment of a $m\times m$ random matrix  ${\bf S}$  refers  to the quantities $\frac{1}{m}\textsf{Tr}\mathbb{E}({\bf S}^r)$ for $r\in\mathbb{Z}$.  
Along this line,  a large amount of works, mainly driven by the recent advances in spectral analysis of large dimensional random matrices, have considered the computation of the asymptotic moments, the term "asymptotic" referring to the  regime in which  the dimensions of the underlying random matrix grow simultaneously large. Among the existing works  in this direction,  we can cite for instance, the work in \cite{couillet-08,kammoun_yao} where the computation of the asymptotic moments is used to infer the transmit power of multiple signal sources, that of \cite{Ryan} dealing with the asymptotic moments of random Vandermonde matrices and finally that of  \cite{jacob}, where the authors studied the asymptotic behavior of the moments in order to allow for the design of a low complexity receiver with a comparable performance to the linear minimum mean square error (LMMSE) detector. 
While working under the asymptotic regime has  enabled the derivation of closed-form expressions for all kind of moments, it presents the drawback of being less accurate for finite dimensions. Alternatively, one might consider the exact approach, which relies on the  already available expression of the marginal eigenvalues' density of Gram random matrices. Interestingly, this approach, despite its seemingly simplicity, has mainly been limited to  computing the moments of Wishart random matrices \cite{maiwald,letac-04}. To the best of the authors' knowledge,  the case of random Gram matrices has never been thoroughly investigated. This lies behind the principal motivation of the present work.  

In this paper, we consider the derivation of the exact moments of random matrices of the form  $\mathbf{S}=\mathbf{H}^*\mathbf{\Lambda}\mathbf{H}$,  where $\mathbf{H}$ is a $n \times m$ ($n > m$) matrix with independent  and identically distributed (i.i.d.) zero-mean unit variance complex Gaussian random entries, and $\mathbf{\Lambda}$ is a fixed $n \times n$ positive definite matrix. It is worth pointing out that matrix ${\bf S}$ cannot be classified as a Wishart random matrix. However, its positive moments $\frac{1}{m}\textsf{Tr} \mathbb{E} {\bf S}^r, r\geq 0$ coincide with those of the Wishart random matrix $\mathbf{\Lambda}^{\frac{1}{2}}\mathbf{H}\mathbf{H}^*\mathbf{\Lambda}^{\frac{1}{2}}$, and  can be thus computed by using existing results on the moments of Wishart matrices. As far as inverse moments are considered $(r<0)$, the same artifice is of no help, mainly because the random matrix $\mathbf{\Lambda}^{\frac{1}{2}}\mathbf{H}\mathbf{H}^*\mathbf{\Lambda}^{\frac{1}{2}}$ becomes singular and thus inverse moments cannot be defined. Besides, from a theoretical standpoint, computing the inverse moments using the  Mellin transform derived in \cite{giusi} is not an easy task. The crude use of the expression provided in \cite{giusi} brings about singularity issues, as will be demonstrated in the course of the paper. Answering to the so-far unsolved question of computing the inverse moments of Gram random matrices constitutes the main contribution of this work.  Additionally, based on the obtained closed-form expression of the exact moments, we revisit some problems in linear estimation. In particular, we provide closed-form expressions of the mean square error for the   best linear unbiased estimator (BLUE) and the linear minimum mean square error estimator (LMMSE) in both high and low SNR regimes. We also derive as a further application the optimal tuning of the windowing factor used in covariance matrix estimation. 
 The remainder of this paper is organized as follows. In Section II, we present the main result of this paper  giving closed form expressions of the inverse moments. In section III, we provide some potential applications and discuss some performance metrics. We then conclude the paper in section IV. Mathematical details are provided in the appendices.\\
\emph{Notations:} 
Throughout this paper, we use the following notations :  $\mathbb{E}\left(\mathbf{X}\right)$ stands for the expectation of a random quantity $\mathbf{X}$ and $\mathbb{E}_{\mathbf{X}}\left(f\right)$ stands for the expected value of $f$ with respect to $\mathbf{X}$.  Matrices are denoted by bold capital letters, rows and columns of the matrices are referred with lower case bold letters ($\mathbf{I}_n$ is the size-$n$ identity matrix). If $\mathbf{A}$ is a given matrix, $\mathbf{A}^t$ and $\mathbf{A}^*$ stand respectively for its transpose and transconjugate. For a square matrix $\mathbf{A}$, we respectively denote by $\textsf{Tr}\left(\mathbf{A}\right)$, $\det \left(\mathbf{A}\right)$ and $\left \|\mathbf{A}\right\|$ its trace, determinant and spectral norm. We refer by $\left[\mathbf{A}\right]_{i,j}$ the $\left(i,j\right)$th entry of $\mathbf{A}$ and by $\textsf{diag}\left(a_1, a_2, \cdots, a_n\right)$ the diagonal matrix with diagonal elements, $a_1, a_2, \cdots, a_n$. 

%

\section{Exact Closed-form expression for the  moments} \label{section2}
Consider a $\left(n \times m\right)$ random matrix $\mathbf{H}$ with i.i.d zero-mean unit variance complex Gaussian random entries with $m<n$. Let $\mathbf{\Lambda}$ be a deterministic $\left(n \times n\right)$ positive definite matrix with distinct eigenvalues $(\theta_1\leq\cdots\leq\theta_n)$ and define the Gram matrix ${\bf S}$ as:
\begin{equation}
\mathbf{S}=\mathbf{H}^*\mathbf{\Lambda}\mathbf{H}.
\label{eq:S}
\end{equation}
In this paper, we consider the computation of the moments $\mu_{\mathbf{\Lambda}}\left(r\right)$ defined as:
\begin{equation} \label{moment}
\mu_{\mathbf{\Lambda}}\left(r\right)\myeq\frac{1}{m}\textsf{Tr}\left(\mathbb{E}_{\mathbf{H}}\{\mathbf{S}^{r}\}\right), \quad r \in \mathbb{Z}.
\end{equation}
 As we will see later, in contrast to the positive moments $(r>0)$ which can be directly obtained from the marginal eigenvalues' density, the derivation of the inverse moments $(r<0)$ is not immediate, requiring a careful analysis of the available existing results.

In the following, we will build on the exact approach in order to derive closed-form expressions for the moments  $\mu_{\mathbf{\Lambda}}\left(r\right)$.  The asymptotic moments will be dealt with subsequently in order to illustrate their inefficiency in evaluating the moments of the eigenvalues of small dimensions Gram matrices.  
\subsection{Closed-form Expressions for the Exact Moments in Fixed Dimensions}

 The exact calculation of moments  is mainly based on existing results on the marginal density of the eigenvalues of $\mathbf{S}$. These results characterize the Mellin transform of the marginal density, the definition of which is given by: 
\begin{mydef}
Denote by $\xi\mapsto f_{\lambda}\left(\xi\right)$ the marginal density distribution of an unordered eigenvalue of $\mathbf{S}$. Then, the Mellin transform of $f_{\lambda}\left(.\right)$ is defined as
\begin{equation} \label{mellin}
\mathcal{M}_{f_{\lambda}}(s) \triangleq\int_{0}^{\infty}\xi^{s-1}f_{\lambda}\left(\xi\right) d\xi.
\end{equation}
\end{mydef} 
With the above definition at hand, we are now in position to recall the following Lemma that provides a closed-form expression for the Mellin Transform of the marginal density of ${\bf S}$:
\begin{lemma}{\cite[Theorem 2]{giusi}}
Let ${\bf S}$ be as in \eqref{eq:S}. Then,
\small
\begin{equation}\label{mellinlemma}
\begin{split}
\mathcal{M}_{f_{\lambda}}(s)&=L\sum
_{j=1}^{m}\sum_{i=1}^{m}\mathcal{D}\left(i,j\right)\Gamma\left(s+j-1\right)\Biggl( \theta_{n-m+i}^{n-m+s+j-2} \\& -\sum_{l=1}^{n-m}\sum_{k=1}^{n-m}\left[\mathbf{\Psi}^{-1}\right]_{k,l}\theta_{l}^{n-m+s+j-2}\theta_{n-m+i}^{k-1}\Biggr)
\end{split}
\end{equation}
\label{lemma:mellin}
\normalsize
with $L=\frac{\det\left(\mathbf{\Psi}\right)}{m\prod_{k <l}^n \left(\theta_l-\theta_k\right)\prod_{l=1}^{m-1}l!}$, $\Gamma(.)$ the Gamma function and $\mathbf{\Psi}$ is the $\left(n-m\right)\times \left(n-m\right)$ Vandermonde matrix,

\[\mathbf{\Psi}=\begin{bmatrix}
1 &\theta_1  &\cdots   &\theta_1^{n-m-1} \\
\vdots  &\vdots   &\ddots   &\vdots  \\
 1&\theta_{n-m}  & \cdots  & \theta_{n-m}^{n-m-1}
\end{bmatrix}\]
where $\mathcal{D}\left(i,j\right)$ is the $(i,j)-$cofactor of the $(m\times m)$ matrix $\mathbf{\mathcal{C}}$ whose $(l,k)-$th entry is given by
\begin{align*}
\tiny
 &(k-1)!\Biggl(\theta_{n-m+l}^{n-m+k-1}
-\sum_{p=1}^{n-m}\sum_{q=1}^{n-m}\left[\mathbf{\Psi}^{-1}\right]_{p,q} \\
& \times \theta_{n-m+l}^{p-1} \theta_{q}^{n-m+k-1}\Biggr).
\end{align*}
\label{lemma:key}
\end{lemma}

The  exact moments  $\mu_{\mathbf{\Lambda}}\left(r\right)$ for $r\geq 0$ can be obtained as a direct consequence of Lemma \ref{lemma:mellin} by replacing in \eqref{mellinlemma} $s$ by $r+1$, thereby yielding the following corollary:
\begin{corollary}
For $r\geq 0$, the moments $\mu_{\mathbf{\Lambda}}\left(r\right)$ are given by:
$$
\mu_{\mathbf{\Lambda}}\left(r\right)=\mathcal{M}_{f_{\lambda}}(r+1)
$$
where $\mathcal{M}_{f_{\lambda}}(r+1)$ is given by \eqref{mellinlemma}.
\end{corollary}
In sharp contrast to the case of positive moments ($r\geq 0$), the inverse moments can not be obtained by a crude substitution of $s$ by $-r-1$ for $r\geq 0$. The problem essentially stems from the terms in the sum wherein the Gamma function is applied to negative integers on which it is not defined.
This might give the impression that the inverse moments are infinite and cannot be thus computed. Such a quick conclusion goes, however, against the existing results on inverse moments available for wishart matrices, thus leading us to suspect  the effect of the Gamma function to be cancelled out in one way or another.  
In order to study the expected compensation effect, it is natural to analyze the behavior of $\mathcal{M}_{f_{\lambda}}(s-r+1)$  for small values of $s$. If a limit exists as $s$ goes to zero, one might expect it to coincide with the sought-for value of the $r$-th moment. 
Such an intuition is confirmed by theory under some conditions on $r$ as it can be shown from the following Lemma.

\begin{lemma}
If $r\geq n-m$, then the limit $
 \lim_{s \downarrow 0} \mathcal{M}_{f_{\lambda}}(s-r+1).
$ exists and
$$
\mu_{\mathbf{\Lambda}}(-r)= \lim_{s \downarrow 0} \mathcal{M}_{f_{\lambda}}(s-r+1).
$$
\label{lemma:limit}
\end{lemma}
\begin{proof}
Let $x\mapsto p(x)$ be the probability density function corresponding to the smallest eigenvalue of ${\bf H}{\bf H}^{*}$. Then, obviously,
\begin{equation}
\mathcal{M}_{f_{\lambda}}(s-r+1) \leq \frac{1}{\theta_1}\int_0^\infty x^{s-r}p(x)dx. 
\label{eq:limit_equation}
\end{equation}
It ensues from the Monotone convergence theorem applied to the sequence of functions $(x\mapsto x^{s-r}p(x))_{s\geq 0}$ and $\left(x\mapsto x^{s-r}f(x)\right)_{s\geq 0}$ that if limits for  the both hand sides of \eqref{eq:limit_equation} exist, they must be equal respectively to $\mu_{\mathbf{\Lambda}}(-r)$ and $\frac{1}{\theta_1}\int_0^\infty x^{-r}p(x)dx$
From Theorem 5.4 in \cite{edelman-phd}, we know that $x^{-r}p(x)$ is integrable provided that $r\leq n-m$. Therefore, for $r\leq n-m$, $\mu_{\mathbf{\Lambda}}(-r)$ is finite and satisfies:
 $$
\mu_{\mathbf{\Lambda}}(-r)=\lim_{s\downarrow 0}\mathcal{M}_{f_{\lambda}}(s-r+1) .
$$
\end{proof}
From Lemma \ref{lemma:limit}, we can see that the computation of the moments $\mu_{\mathbf{\Lambda}}(-r)$ amounts to evaluating the limit of $\mathcal{M}_{f_{\lambda}}(s-r+1)$ as $s$ goes to zero. 
A careful scrutiny of the expression of  $\mathcal{M}_{f_{\lambda}}(s-r+1)$ reveals that the sum over index $j$ makes appear two types of terms. The first type corresponds to those indices of $j$ for which $-r+j-1$ is positive. The limits of these  terms can be computed normally by setting $s$ to $0$ since  $\Gamma(-r+j-1)$ is properly defined. The second type of terms is more difficult to analyze, since it corresponds to those indices of $j$ for which $-r+j-1$ is negative. Obviously, these two types of terms cannot be handled similarly. In light of this observation, it is sensible to decompose $\mathcal{M}_{f_{\lambda}}(s-r+1)$ as the sum of two quantities depending on the value of $j$, whether it is below or above $r+1$. This decomposition writes as:
\begin{equation} \label{momentlimit}
\mathcal{M}_{f_{\lambda}}(s-r+1) 
= \mathcal{M}_1\left(s-r+1\right) +\mathcal{M}_2\left(s-r+1\right) ,
\end{equation} 
where \begin{align*}
\mathcal{M}_1\left(s\right)&=L\sum
_{j=1}^{r}\sum_{i=1}^{m}\mathcal{D}\left(i,j\right)\Gamma\left(s+j-1\right)\Biggl( \theta_{n-m+i}^{n-m+s+j-2} \\& -\sum_{l=1}^{n-m}\sum_{k=1}^{n-m}\left[\mathbf{\Psi}^{-1}\right]_{k,l}\theta_{l}^{n-m+s+j-2}\theta_{n-m+i}^{k-1}\Biggr) \\
\mathcal{M}_2\left(s\right)&=L\sum
_{j=r+1}^{m}\sum_{i=1}^{m}\mathcal{D}\left(i,j\right)\Gamma\left(s+j-1\right)\Biggl( \theta_{n-m+i}^{n-m+s+j-2} \\& -\sum_{l=1}^{n-m}\sum_{k=1}^{n-m}\left[\mathbf{\Psi}^{-1}\right]_{k,l}\theta_{l}^{n-m+s+j-2}\theta_{n-m+i}^{k-1}\Biggr).
\end{align*}
We will first handle the second term $\mathcal{M}_2\left(s-r+1\right)$, gathering indices $j$ for which $-r+j-1$ is positive. Interestingly, we can prove that its limit is zero as $s\downarrow0$, which shows that it does not contribute in the expression of the final moment. 
\begin{proposition} The term $\mathcal{M}_2\left(s-r+1\right)$ vanishes as $s$ goes to zero i.e, 
\begin{align*}
\lim_{s \rightarrow 0} \mathcal{M}_2\left(s-r+1\right)=0, \quad r=1, \cdots, m.
\end{align*}
\end{proposition}
\begin{proof}
See Appendix A for the proof.
\end{proof}
The expression of the moment is thus totally ruled out by the contribution of the first term $\mathcal{M}_1\left(s-r+1\right)$. Before providing the expression of its limit as $s\downarrow 0$, we shall introduce the following notations:

\begin{align*}
\mathbf{a}_{j}&=\left[\theta_1^{n-m-r+j-1}, \theta_2^{n-m-r+j-1}, \cdots, \theta_{n-m}^{n-m-r+j-1} \right]^t\\
\mathbf{D}_i&=\textsf{diag}\Biggl(\log\left(\frac{\theta_{n-m+i}}{\theta_1}\right), \log\left(\frac{\theta_{n-m+i}}{\theta_2}\right), \\
&\cdots,\log\left(\frac{\theta_{n-m+i}}{\theta_{n-m}}\right)\Biggr) \\
\mathbf{b}_i&\triangleq\left[1,\theta_{n-m+i}, \cdots, \theta_{n-m+i}^{n-m-1}\right]^t .
\end{align*}
With these notations at hand, we are now in position to state the following result:
\begin{proposition} Let $p=\min\left(m,n-m\right)$, then for $1 \leq r \leq p$ we have 
\small
\begin{align*}
&\lim_{s \rightarrow 0}\mathcal{M}_1\left(s-r+1\right)\\
&=L\sum
_{j=1}^{r}\sum_{i=1}^{m}\mathcal{D}\left(i,j\right)\frac{\left(-1\right)^{r-j}}{\left(r-j\right)!}\mathbf{b}_i^t\mathbf{\Psi}^{-1} \mathbf{D}_i\mathbf{a}_j .
\end{align*}
\normalsize
\end{proposition}
\begin{proof}
See Appendix B for a detailed proof.
\end{proof} 
Combining the findings of the above propositions, we finally obtain  the following result: 
\begin{theorem}
For $1 \leq r \leq p$, we have
\footnotesize
\begin{align*}
\mu_{\mathbf{\Lambda}}\left(-r\right)=L\sum
_{j=1}^{r}\sum_{i=1}^{m}\mathcal{D}\left(i,j\right)\frac{\left(-1\right)^{r-j}}{\left(r-j\right)!}\mathbf{b}_i^t\mathbf{\Psi}^{-1} \mathbf{D}_i\mathbf{a}_j .
\end{align*}
\normalsize
\end{theorem}
\begin{remark}
Without loss of generality, we can easily show that  matrix $\mathbf{\Lambda}$ can be considered as diagonal with diagonal elements $\theta_1,\cdots,\theta_N$. This can be seen from  the eigendecomposition of $\mathbf{\Lambda}$ as follows 
\begin{equation} \label{eigen}
\mathbf{\Lambda}=\mathbf{U}^*\mathbf{D}\mathbf{U},
\end{equation}
where $\mathbf{D}=\textsf{diag}\left(\theta_1, \theta_2, \cdots, \theta_n\right)$ and $\mathbf{U}$ is a unitary matrix, i.e. $\mathbf{U}^*\mathbf{U}=\mathbf{U}\mathbf{U}^*=\mathbf{I}_n$. 
Then,
\begin{align*}
\mathbf{S}&=\mathbf{H}^*\mathbf{U}^*\mathbf{D}\mathbf{U}\mathbf{H}\\
&=\left(\mathbf{U}\mathbf{H}\right)^*\mathbf{D}\left(\mathbf{U}\mathbf{H}\right) \\
&=\mathbf{G}^* \mathbf{D}\mathbf{G},
\end{align*} 
where $\mathbf{G}=\mathbf{U}\mathbf{H}$. Since the wishart distribution is unitarily invariant, $\mathbb{\bf G}$ has the same distribution as ${\bf H}$. Therefore,
\begin{equation} \label{moment}
\mu_{\mathbf{\Lambda}}\left(r\right)=\frac{1}{m}\textsf{Tr}\left(\mathbb{E}_{\mathbf{G}}\{\left(\mathbf{G}^* \mathbf{D}\mathbf{G}\right)^{r}\}\right), \quad r \in \mathbb{Z}
\end{equation}
\end{remark}
\subsection{Asymptotic Inverse Moments}
It is well-known from standard results on random matrix theory that   moments of Gram random matrices, can be well-approximated, when properly scaled and for $m$ and $n$ large enough,  by deterministic quantities. However, the derivation of these deterministic approximations differs from the exact approach in several respects, namely it does not rely on the same tool of the Mellin transform  and does not necessarily yield  simple closed-form expressions for any high order moment. As a matter of fact, it is shown in \cite{jacob} that except the special case of $\boldsymbol{\Lambda}$ coinciding with the identity matrix, the computation of higher positive order moments has to be performed iteratively. This also holds  for the case of asymptotic inverse moments, which can be derived using the same approach as in \cite{jacob}.  
This represents the main goal of this section, for which details will be provided for sake of completeness.
 The obtained asymptotic moments will be compared later with the exact ones derived in the previous section. 

Prior to stating the main algorithm leading to the asymptotic inverse moments, we shall first review the following results from random matrix theory.  

\begin{mydef} (Empirical Spectral Distribution) Let $\mathbf{A} \in \mathbb{C}^{m \times m}$ be a Hermitian matrix with eigenvalues $\lambda_1, \lambda_2, \cdots, \lambda_m$. The empirical spectral distribution $F^{\mathbf{A}}$ of $\mathbf{A}$ is defined as 
\begin{equation}
F^{\mathbf{A}}\left(x\right)=\frac{1}{m}\sum_{i=1}^m \mathbf{1}\left(\lambda \leq x\right)
\end{equation}
\end{mydef}
Working directly on the empirical distribution function $F^{{\bf A}}$ is in general a tedious task. Instead, a characterization of its  Stieltjes transform is often considered. The Stieltjes transform corresponding to the empirical distribution $F^{{\bf A}}$ is defined as:
\begin{mydef}(Stieltjes Transform) \label{Stieltjesdef}
For a hermitian matrix  $\mathbf{A}$, the Stieltjes transform is defined as
\begin{equation}
\label{stieltjes}
\begin{split}
\hat{m}_{\mathbf{A}}\left(z\right)& \triangleq \int \frac{1}{\lambda-z}dF^{\mathbf{A}}\left(\lambda\right) \\
&= \frac{1}{m}\textsf{Tr}\left(\mathbf{A}-z \mathbf{I}_m\right)^{-1}
\end{split}
\end{equation}
\end{mydef}
From definition \ref{Stieltjesdef}, it is easy to prove
\begin{equation}
\left(\frac{\partial^k \hat{m}_{\mathbf{A}}\left(z\right) }{\partial z^k}\right)_{z=0}=\frac{k!}{m}\textsf{Tr}\left(\mathbf{A}^{-\left(k+1\right)}\right)
\end{equation}

Let $\mathbf{D}$ be the diagonal matrix as defined in (\ref{eigen}). Then, the Stieltjes Transform (ST) of the empirical measure of $\frac{1}{m}{\bf S}$ converges to a deterministic measure whose ST $\underline{m}\left(z\right)$ is the solution of the following fixed point equation \cite{silverstein}:
\begin{equation} \label{mmoment0}
\underline{m}\left(z\right)=\frac{1}{-z+\frac{1}{m}\sum_{k=1}^n\frac{\left[\mathbf{D}\right]_{k,k}}{1+\left[\mathbf{D}\right]_{k,k}\underline{m}\left(z\right)}}.
\end{equation}
Denote by $\underline{m}^{(r)}$ the $r$-th derivative of $\underline{m}(z)$ at $z=0$.   Along the same arguments as in \cite{jacob}, we can prove that $\underline{m}^{(r)}$ is a consistent estimate of $r!m^r \textsf{Tr}\left(\mathbf{H}^*\boldsymbol{\Lambda}\mathbf{H}\right)^{-\left(r+1\right)}$. This suggests in particular estimating the scaled inverse moments  $\frac{1}{m}\textsf{Tr}(\left(\frac{1}{m}{\bf S}\right)^{-r})$ by $\frac{1}{r!}\underline{m}^{(r)}$. Closed-form expressions for the derivatives of $\underline{m}^{(r)}$ do not exist, but  they can be numerically computed recursively using the result of the following Theorem.
\begin{theorem}
Let $p \geq 1$ and $f_k\left(z\right)=-\frac{1}{1+\left[\mathbf{D}\right]_{k,k}\underline{m}\left(z\right)}$. 
Denote by $f_k^{(p)}$ the $p$-th derivative of $f_k(z)$ at $z=0$. Then, the following relations hold true:
\begin{equation} \label{relation1}
\begin{split}
 p\underline{m}^{\left(p-1\right)} &+\frac{\underline{m}^{\left(p\right)}}{m}\sum_
{k=1}^n \frac{\left[\mathbf{D}\right]_{k,k}f^{\left(0\right)}_k}{1+\left[\mathbf{D}\right]_{k,k}m\left(0\right)} \\
&+\frac{1}{m}\sum_{k=1}^n\sum_{l=1}^{p-1}\binom {p} {l}\frac{\left[\mathbf{D}\right]_{k,k}\underline{m}^{\left(l\right)}f^{\left(p-l\right)}_k}{1+\left[\mathbf{D}\right]_{k,k}\underline{m}\left(0\right)}=0,
\end{split}
\end{equation}
\begin{align} \label{relation2}
f^{\left(p\right)}_k+\frac{\left[\mathbf{D}\right]_{k,k}\underline{m}^{\left(p\right)}f^{\left(0\right)}_k}{1+\left[\mathbf{D}\right]_{k,k}\underline{m}\left(0\right)}+\sum_{l=1}^{p-1}\binom {p} {l}\frac{\left[\mathbf{D}\right]_{k,k}\underline{m}^{\left(l\right)}f^{\left(p-l\right)}_k}{1+\left[\mathbf{D}\right]_{k,k}\underline{m}\left(0\right)}=0.
\end{align}
\end{theorem}
\begin{proof}
See Appendix C for detailed proof.
\end{proof}
Based on the previous theorem, an algorithm can be provided in order to recursively compute the higher order derivatives, $\underline{m}^{\left(k\right)}$. These values will thus immediately serve to compute the deterministic approximations for the moments.

\begin{algorithm}
\caption{Asymptotic inverse moments computation}
\label{CHalgorithm}
\begin{algorithmic}[1]
\\ Compute $\underline{m}\left(0\right)$ using (\ref{mmoment0})
\\ Compute $f_k\left(0\right)=-\frac{1}{1+\left[\mathbf{D}\right]_{k,k}\underline{m}\left(0\right)}$ 
\For {$i=1 \to p$}
 \\ compute $\underline{m}^{\left(i\right)}$ using (\ref{relation1})
 \\ compute $f^{\left(i\right)}_k$ using (\ref{relation2})
\EndFor
\end{algorithmic}
\end{algorithm}

\subsection{Numerical Examples}
We validate the theoretical result stated in Theorem 1 for different values of $m$ and $n$. In particular we compare $\mu_{\mathbf{\Lambda}}\left(-r\right)$ with the normalized asymptotic moments $\frac{\underline{m}^{(r)}}{r!m^{r+1}}$ and the empirical moments obtained by Montecarlo simulations. \\
We start by verifying the result in Figure 1, in the case where $\mathbf{\Lambda}$ is a correlation matrix having the following structure 
\begin{equation}
\label{scatter}
\left[\mathbf{\Lambda}\right]_{i,j}=J_0\left(\pi\left|i-j\right|^2\right),
\end{equation}
where $J_0\left(.\right)$ is the zero-order Bessel function of the first kind. This kind of matrices is used to model the correlation between transmit antennas in a dense scattering environment. 
For simulations, we set $m=3$ and vary $n$ such that $n > m$. \\
In Figure 2, we compare the same quantities in the case where $\mathbf{\Lambda}$ is a random positive definite matrix generated as follows\footnote{We use such matrices to make sure that the obtained results are applicable for broad class of positive definite matrices.}
\begin{equation} \label{lambdarandom}
\mathbf{\Lambda}=\mathbf{I}_n+\mathbf{W}^*\mathbf{W},
\end{equation}
where $\mathbf{W}$ is a $\left(n \times n\right)$ matrix with i.i.d zero-mean unit variance complex Gaussian random entries.
\begin{figure}[t!]
\label{blueperformance}
        \centering
    \includegraphics[width=3.5in]{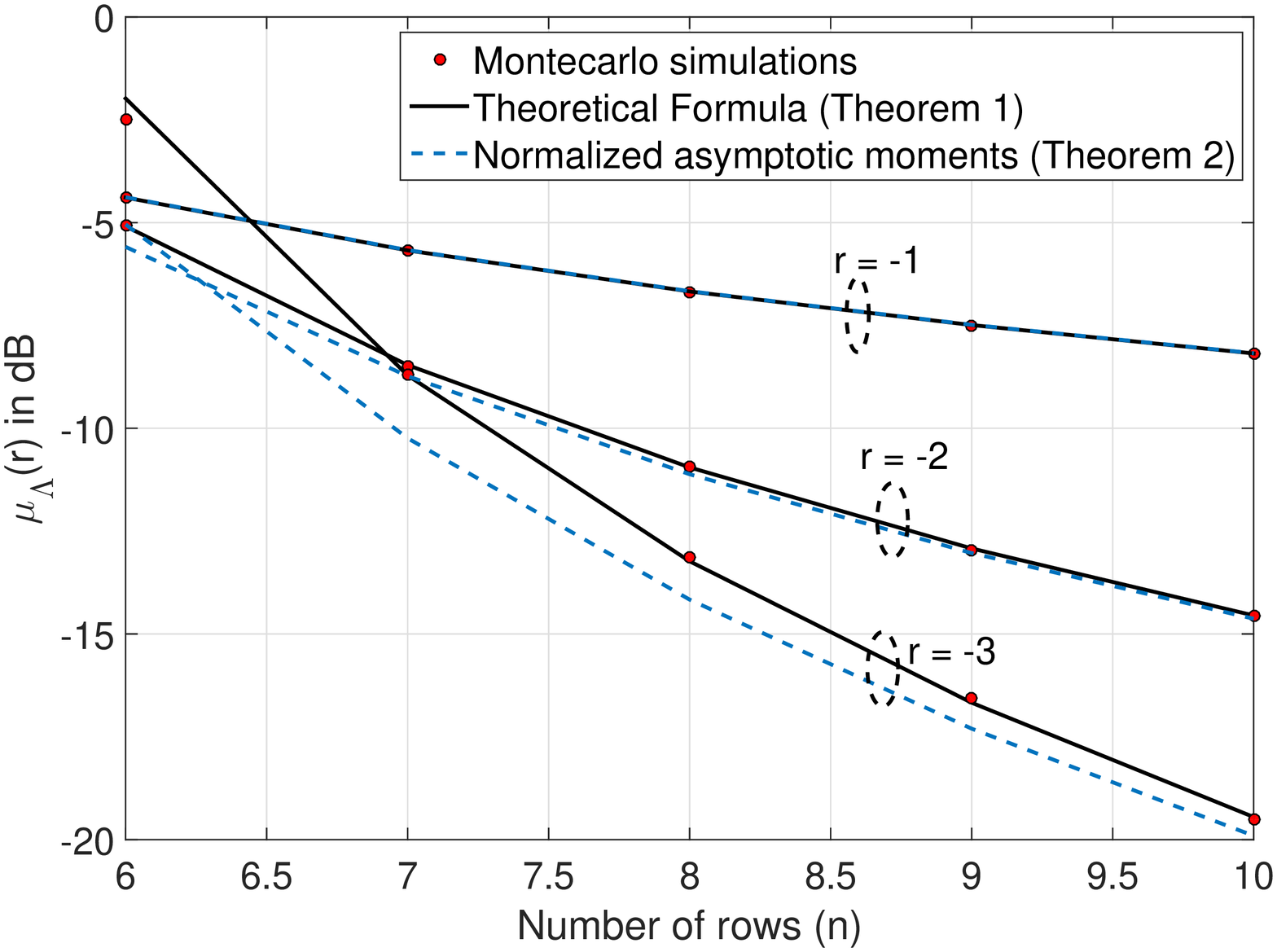}
\caption{Inverse moments for $\mathbf{\Lambda}$ defined as in (\ref{scatter}): A comparison between theoretical result (Theorem 1), normalized asymptotic moments (Theroem 2) and  Montecarlo simulations ($10^4$ realizations).}
\label{fig7}
\end{figure}
\begin{figure}[t!]
\label{blueperformance}
        \centering
    \includegraphics[width=3.5in]{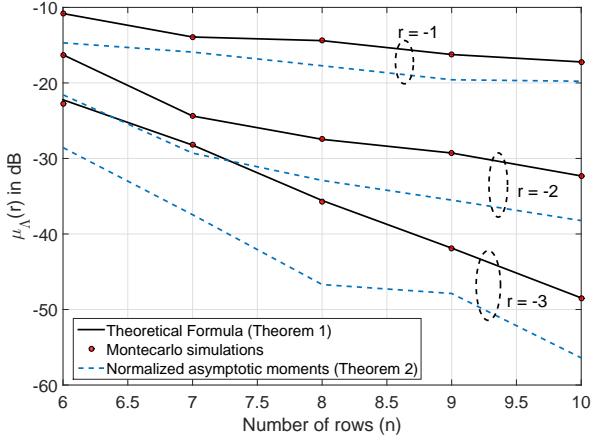}
\caption{Inverse moments for $\mathbf{\Lambda}$ defined as in (\ref{lambdarandom}): A comparison between theoretical result (Theorem 1), normalized asymptotic moments (Theroem 2) and  Montecarlo simulations ($10^4$ realizations).}
\label{fig7}
\end{figure} \\
In both Figures 1 and 2, the theoretical result of Theorem 1 perfectly matches the montecarlo simulations, however the normalized asymptotic moments derived in Theorem 2 only matches the exact results for the case where $r=-1, -2$ where the model in (\ref{scatter}) is adopted. This  can be explained by the fact that the use of normalized asymptotic moments lead to inaccurate results  in the regime of fixed $m$ and $n$  and that the result of Theorem 1 is more suitable in this case. 

\section{Applications of The Inverse moments} 

The computation of inverse moments of  one-sided correlated Gram matrices is paramount to many applications of signal processing. For sake of illustration, we will discuss in the sequel applications of our results to the fields of linear estimation and covariance matrix estimation. 
\subsection{Linear Estimation}
The problem of estimating an unknown signal from a sequence of observations has been widely studied in the literature \cite{kailath,sayed,vpoor} and can be solved if joint statistics or cross correlations of the unknown signal and the observations vector are available. In this line, linear models are a special case where, the input and the output are linearly related as 
\begin{equation} \label{linsys}
\mathbf{y}=\mathbf{H}\mathbf{x}+\mathbf{z},
\end{equation}
where $\mathbf{y} \in \mathbb{C}^{n \times 1}$ is the observations vector, $\mathbf{H}\in \mathbb{C}^{n \times m}$ the channel matrix, $\mathbf{x} \in \mathbb{C}^{m \times 1}$ the unknown signal with covariance matrix $\mathbf{\Sigma}_x$ and $\mathbf{z} \in \mathbb{C}^{n \times 1}$ the noise vector with covariance matrix $\mathbf{\Sigma}_z$. As stated earlier, in order to recover $\mathbf{x}$, joint statistics are required. However, acquiring joint statistics is generally a difficult task either because of the unknown nature of the signal or simply because of the unavailability of the statistics. To overcome this issue,  linear estimators can be viewed as a good alternative.  They are merely based on applying a linear transformation to the observation vector. 
 Obviously, this is a sub-optimal strategy in regards of the minimization of  the \textit{mean square error}, but it is more tractable and permits to explicitly analyze  performances. 
In Table I, we review the explicit expressions of the unknown signal for different estimation techniques 
\begin{table*}[t!] \label{tableestimators}
\caption{Linear Estimation Techniques depending on the available Statistics}
\begin{center}
\begin{tabular}{|p{2cm}||p{1.6cm}||p{4.5cm}||p{3.5cm}|}
  \hline
  \rowcolor{lightgray}
  Linear estimator & Required Statistics & Estimated signal, $\hat{\mathbf{x}}$ & Error covariance matrix, $\mathbb{E}\left(\mathbf{x}-\hat{\mathbf{x}}\right)\left(\mathbf{x}-\hat{\mathbf{x}}\right)^*$ \\\hline
LS & $\emptyset$ & $\left(\mathbf{H}^*\mathbf{H}\right)^{-1}\mathbf{H}^*\mathbf{y}$ &  $\left(\mathbf{H}^*\mathbf{H}\right)^{-1}$\\\hline
  BLUE& $\mathbf{\Sigma}_z$ & $\left(\mathbf{H}^*\mathbf{\Sigma}_z^{-1}\mathbf{H}\right)^{-1}\mathbf{H}^*\mathbf{\Sigma}_z^{-1}\mathbf{y}$ & $\left(\mathbf{H}^*\mathbf{\Sigma}_z^{-1}\mathbf{H}\right)^{-1}$\\\hline
 LMMSE  & $\mathbf{\Sigma}_z$, $\mathbf{\Sigma}_x$ & $\left(\mathbf{\Sigma}_x^{-1}+\mathbf{H}^*\mathbf{\Sigma}_z^{-1}\mathbf{H}\right)^{-1}\mathbf{H}^*\mathbf{\Sigma}_z^{-1}\mathbf{y}$ & $\left(\mathbf{\Sigma}_x^{-1}+\mathbf{H}^*\mathbf{\Sigma}_z^{-1}\mathbf{H}\right)^{-1}$\\
  \hline
\end{tabular}
\label{TableSNR}
\end{center}
\end{table*}
depending on the informations available about the signal and the noise statistics. 
In what follows, we make the following assumptions:
\begin{itemize}
\item $\mathbf{H}$ is a $\left(n \times m\right)$ matrix with \textit{i.i.d} complex zero mean unit variance Gaussian random entries
\item $\mathbf{z}$ is a $\left(n \times 1\right)$  zero mean additive Gaussian noise with covariance matrix $\mathbf{\Sigma}_z=\mathbb{E}\{\mathbf{z}\mathbf{z}^*\}$, i.e. \\ $\mathbf{z} \sim \mathcal{CN}\left(\mathbf{0}_{n},\mathbf{\Sigma}_z\right)$.
\end{itemize}

\subsubsection{An Exact expression for The BLUE Average Estimation Error}\label{sysmod}

Let $n > m$ and consider the same linear system as in (\ref{linsys}). With the the noise covariance matrix $\mathbf{\Sigma}_z$ at hand, the best linear unbiased estimator (BLUE) \cite{sayed} recovers $\mathbf{x}$ as:
\begin{equation}\label{blue}
\begin{split}
\hat{\mathbf{x}}_{blue}&=\left(\mathbf{H}^*\mathbf{\Sigma}_z^{-1}\mathbf{H}\right)^{-1}\mathbf{H}^*\mathbf{\Sigma}_z^{-1}\mathbf{y} \\
&=\mathbf{x}+\left(\mathbf{H}^*\mathbf{\Sigma}_z^{-1}\mathbf{H}\right)^{-1}\mathbf{H}^*\mathbf{\Sigma}_z^{-1}\mathbf{z} \\
&=\mathbf{x}+\mathbf{e}_{blue},
\end{split}
\end{equation}
where $\mathbf{e}_{blue}=\left(\mathbf{H}^*\mathbf{\Sigma}_z^{-1}\mathbf{H}\right)^{-1}\mathbf{H}^*\mathbf{\Sigma}_z^{-1}\mathbf{z}$ is the residual error after applying the BLUE. We denote by $\mathbf{\Sigma}_{e,blue}=\mathbb{E}\{\mathbf{e}_{blue}\mathbf{e}_{blue}^*\}$ the covariance matrix of $\mathbf{e}_{blue}$, then $\mathbf{\Sigma}_{e,blue}=\left(\mathbf{H}^*\mathbf{\Sigma}_z^{-1}\mathbf{H}\right)^{-1}$. Using the result of Theorem 1, the average estimation error is thus given by:
\begin{equation}
\label{avgerrorblue}
\begin{split}
\mathbb{E}_{\mathbf{H}}\{ \Vert\hat{\mathbf{x}}_{blue}-\mathbf{x}\Vert^2\}&=\mathbb{E}_{\mathbf{H}}\textsf{Tr}\left(\mathbf{\Sigma}_{e,blue}\right) \\
&=\mathbb{E}_{\mathbf{H}}\textsf{Tr}\left(\left(\mathbf{H}^*\mathbf{\Sigma}_z^{-1}\mathbf{H}\right)^{-1}\right) 
\\&=m\mu_{\mathbf{\Lambda}}\left(-1\right),
\end{split}
\end{equation}
where $\mathbf{\Lambda}=\mathbf{\Sigma}_z^{-1}$. \\
For simulation purposes, we set $m=3$, and consider $\mathbf{\Lambda}$ as in (\ref{scatter}) and (\ref{lambdarandom}).  
Then, we compare the empirical average estimation error using Montecarlo simulaton for different values of $n$ with the theoretical result derived in Theorem 1. 
As shown in Figure 3, the theoretical performance exactly matches the exact performance of the BLUE in terms of average estimation error for both models of $\mathbf{\Lambda}$ in (\ref{scatter}) and (\ref{lambdarandom}).
\begin{figure}[t!]
\label{blueperformance}
        \centering
    \includegraphics[width=3.5in]{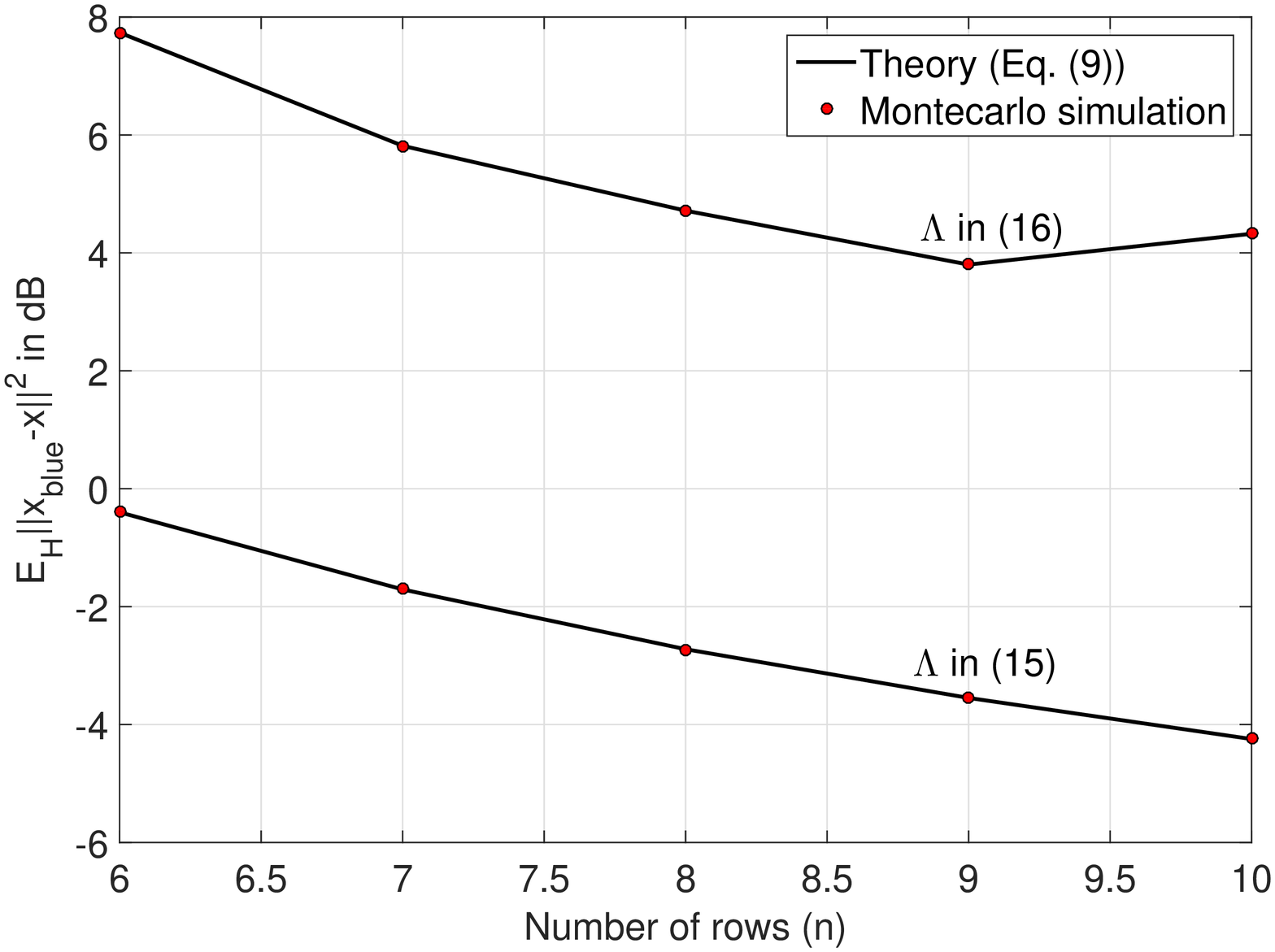}
\caption{BLUE average estimation error for $m=3$: Montecarlo simulation ($10^4$ realizations) versus theory (Theorem 1)}
\label{fig7}
\end{figure}
\subsubsection{Approximation of the LMMSE average estimation error}
Consider the linear system as in (\ref{linsys}), where we assume additionally that the covariance matrix of the unknown signal $\mathbf{x}$ is known and given by $\mathbf{\Sigma}_x$.  The linear minimum mean square error estimate (LMMSE) of $\mathbf{x}$ is thus given by:
\begin{equation}
\hat{\mathbf{x}}_{lmmse}=\left(\mathbf{\Sigma}_x^{-1}+\mathbf{H}^*\mathbf{\Sigma}_z^{-1}\mathbf{H}\right)^{-1}\mathbf{H}^*\mathbf{\Sigma}_z^{-1}\mathbf{y}
\end{equation}
Consequently, the estimation error can be calculated as 
\begin{align}
\mathbf{e}_{lmmse}=\left(\mathbf{\Sigma}_x^{-1}+\mathbf{H}^*\mathbf{\Sigma}_z^{-1}\mathbf{H}\right)^{-1}\mathbf{H}^*\mathbf{\Sigma}_z^{-1}\mathbf{z}
\end{align}
By  standard computations and based on the result derived in \cite{sayed}, the error covariance matrix is given by
\begin{equation} \label{covarianceMMSE}
\mathbf{\Sigma}_{e,lmmse}=\left(\mathbf{\Sigma}_x^{-1}+\mathbf{H}^*\mathbf{\Sigma}_z^{-1}\mathbf{H}\right)^{-1}
\end{equation}
Therefore, the average estimation error for the LMMSE estimator is given by 
\begin{equation}
\begin{split}
\mathbb{E}_{\mathbf{H}}\{ \Vert\hat{\mathbf{x}}_{lmmse}-\mathbf{x}\Vert^2\}&=\mathbb{E}_{\mathbf{H}}\textsf{Tr}\left(\mathbf{\Sigma}_{e,lmmse}\right) \\
&=\mathbb{E}_{\mathbf{H}}\textsf{Tr}\left(\left(\mathbf{\Sigma}_x^{-1}+\mathbf{H}^*\mathbf{\Sigma}_z^{-1}\mathbf{H}\right)^{-1}\right) 
\end{split}
\end{equation}
Evaluating the LMMSE average estimation error is in general very difficult, however, for the simple case where $\mathbf{\Sigma}_x=\sigma_x^2 \mathbf{I}_{n}$, it is possible to obtain an approximation depending on the value of $\sigma_x^2$. In the following theorem, we provide tight approximations for the LMMSE average estimation error for the cases: $\sigma_x^2 \gg 1$ (High SNR regime) and $\sigma_x^2 \ll 1$ (Low SNR regime).
\begin{theorem}\label{approximation}
Let $\mathbf{\Lambda}=\mathbf{\Sigma}_z^{-1}$. Then, the LMMSE average estimation error at both the high SNR regime ($\sigma_x^2 \gg 1$) and the low SNR regime ($\sigma_x^2 \ll 1$) is  given by 
\begin{enumerate}
\item High SNR regime:
\begin{equation}\label{theorem2}
\begin{split}
&\mathbb{E}_{\mathbf{H}}\{ \Vert\hat{\mathbf{x}}_{lmmse}-\mathbf{x}\Vert^2\}\\
&=m\sum_{k=0}^{l} \frac{\left(-1\right)^k}{\sigma_x^{2k}} \mu_{\mathbf{\Lambda}}\left(-k-1\right)+o\left(\sigma_x^{-2r}\right)
\end{split}
\end{equation}
where $l\leq p-1$ with $p=\min(m,n-m)$. 
\item Low SNR regime:
\begin{align}\label{theorem2}
\mathbb{E}_{\mathbf{H}}\{ \Vert\hat{\mathbf{x}}_{lmmse}-\mathbf{x}\Vert^2\}=m\sum_{k=0}^{\infty} \left(-1\right)^k\sigma_x^{2k+2} \mu_{\mathbf{\Lambda}}\left(k\right)
\end{align}
\end{enumerate}

\end{theorem}
\begin{proof}
See Appendix D for Proof.
\end{proof}
\begin{figure}[t!]
        \centering
    \includegraphics[width=3.5in]{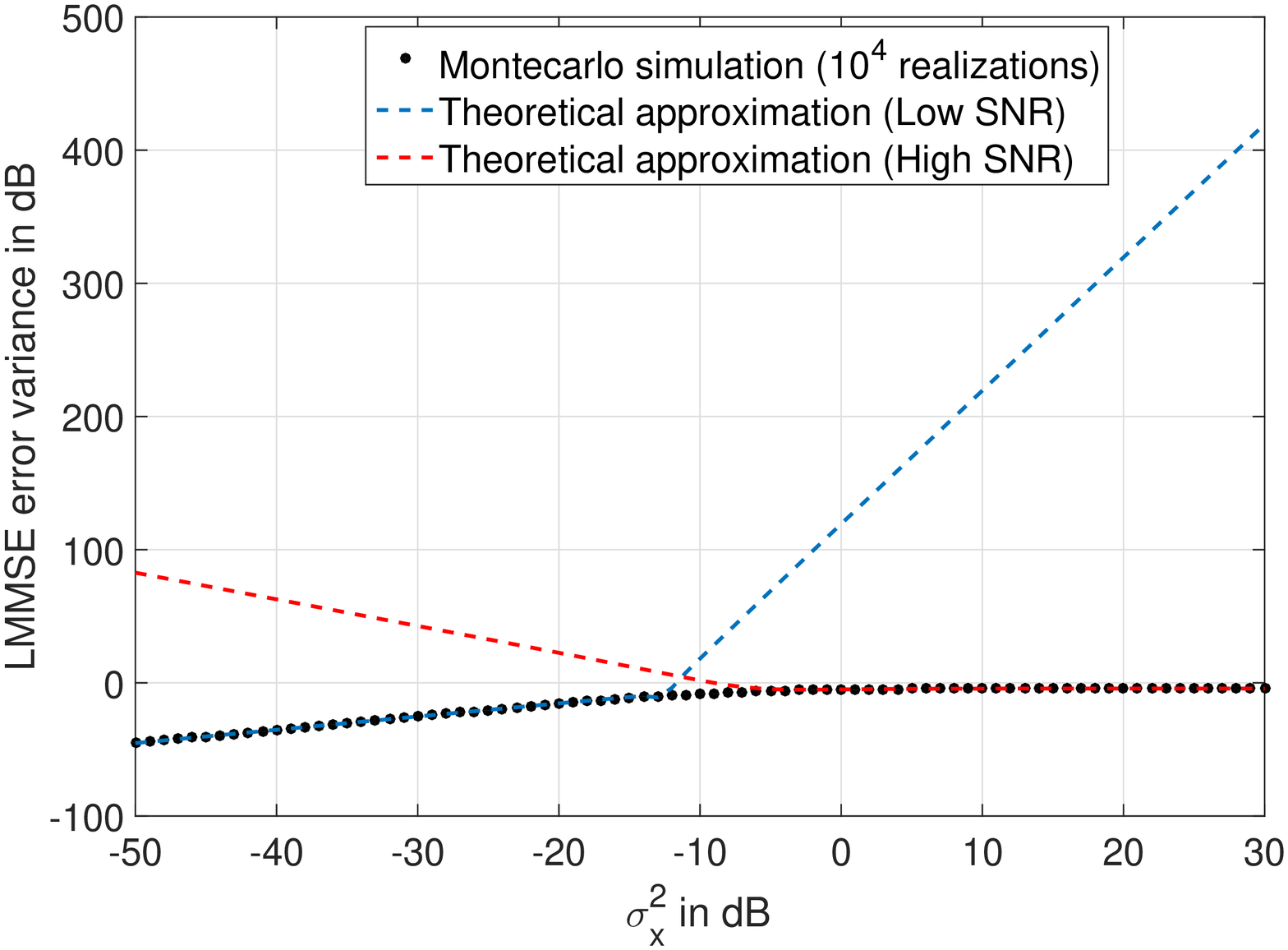}
\caption{LMMSE mean square error with $\mathbf{\Sigma}_z$ as in (\ref{scatter}): Montecarlo simulation versus theoretical approximation for the low and high SNR regimes.}
\label{fig7}
\end{figure} 
\begin{figure}[t!]
        \centering
    \includegraphics[width=3.5in]{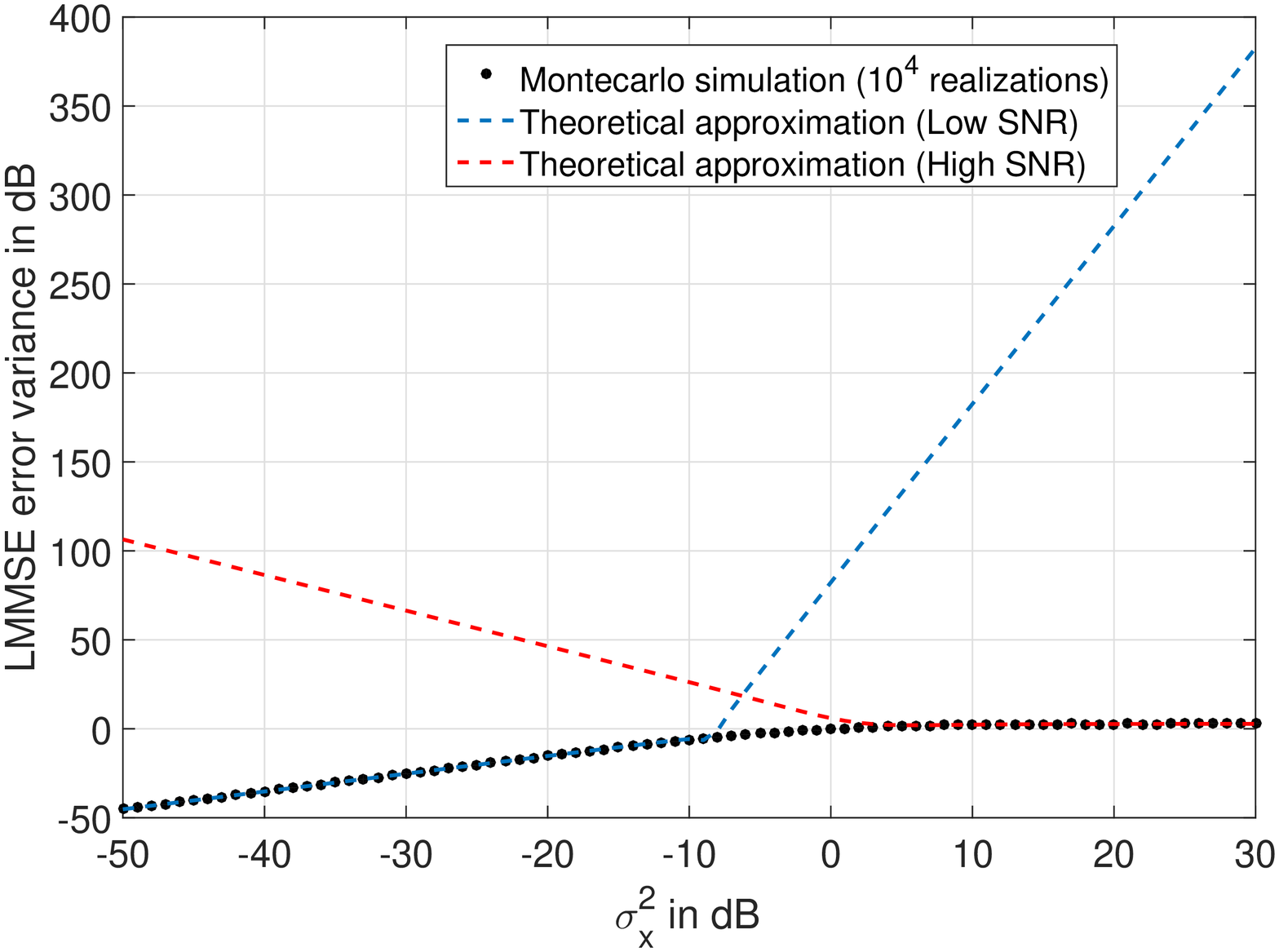}
\caption{LMMSE mean square error with $\mathbf{\Sigma}_z$ as in (\ref{lambdarandom}): Montecarlo simulation versus theoretical approximation for the low and high SNR regimes.}
\label{fig7}
\end{figure} 
To validate the approximations derived in Theroem 3, we set $m=3$ and $n=10$ and apply the obtained results for both correlation models in (\ref{scatter}) and (\ref{lambdarandom}).
For that in Figures 4 and 5, we compare the mean square error of the LMMSE using Montecarlo simulations with the approximations derived in Theorem 3. 
As shown in Figures 4 and 5, the approximation is quiet tight in both high and low SNR regimes and almost cover the whole SNR range.
\subsection{Sample Correlation Matrix (SCM)}
Estimation of covariance matrices is  of fundamental importance to several adpative processing applications. Assume that the measurements are arranged into an input vector $\mathbf{u} \in \mathbb{C}^{m \times 1}$ called also the observation vector. If the input process is zero-mean, its covariance matrix is given by:  
\begin{equation}
\mathbf{R}\triangleq \mathbb{E}\{\mathbf{u}\mathbf{u}^{*}\},
\end{equation}
where the expectation is taken over all realization of the input. The covariance matrix $\mathbf{R}$ is usually unknown, and thus has to be estimated. Assuming the input process to be ergodic, the covariance matrix can be estimated via time averaging. A well-known estimator is the sample correlation matrix (SCM) which is  given by \cite{MITthesis}
\begin{equation}
\hat{\mathbf{R}}\left(n\right)=\frac{1}{n}\sum_{k=1}^n \mathbf{u}\left(k\right)\mathbf{u}^*\left(k\right),
\end{equation}
This is called rectangularly windowed SCM, where $\mathbf{u}\left(k\right)$ is the input vector at discrete time $k$ and $n$ is the length of the observation window. When the observations are Gaussian distributed, the SCM  is the maximum likelihood (ML) estimator of the correlation matrix \cite{estimationbook}. Moreover, for a fixed and finite input size $m$, as the window size $n \rightarrow \infty$, the SCM converges to the input correlation matrix \cite{couillet}, in the sense that:
\begin{equation}
\left \| \mathbf{R}-\hat{\mathbf{R}}\left(n\right) \right \| \rightarrow 0, \quad a.s.
\end{equation} 
where $\left \|. \right \|$ is a spectral norm of a matrix.\\
However, the number of measurement is usually finite for practical applications. Thus, it is for a practical interest to evaluate the performance of the SCM when the window size is finite. In order to measure the accuracy of the estimator,  we define the average \textit{loss} as the average distance between the input correlation matrix and its estimated version using SCM for a given window size $n$ \cite{yang-94}: 
\begin{equation} \label{loss}
Loss\left(n\right)\triangleq \mathbb{E}\left \| \mathbf{R}^{\frac{1}{2}}\hat{\mathbf{R}}^{-1}\left(n\right)\mathbf{R}^{\frac{1}{2}}-\mathbf{I}_m\right \|^2_F,
\end{equation}
where $\mathbf{R}^{\frac{1}{2}}$ is a positive semi-definite square root of $\mathbf{R}$ and $\left \|. \right \|_F$ is the Forbenius norm of a matrix. \\
In order to emphasize some measurements relevant for the estimation of the correlation matrix, an exponentially weighted SCM can be used and it is given by \cite{MITthesis}:
\begin{equation} \label{weighted}
\hat{\mathbf{R}}\left(n\right)=(1-\lambda)\sum_{k=1}^n \lambda^{n-k}\mathbf{u}\left(k\right)\mathbf{u}^*\left(k\right),
\end{equation}
where $\lambda^{n-k}$ is the weight associated to the measurement vector at time instant $k$, the coefficient $\lambda$ being the forgetting factor. In the case where $\mathbf{u}\left(k\right)$ is modeled as a colored process $\mathbf{u}\left(k\right)=\mathbf{R}^{\frac{1}{2}}\mathbf{x}\left(k\right)$, where $\mathbf{x}\left(k\right) \in \mathbb{C}^{m \times 1}$ is a vector of $i.i.d$ Gaussian zero mean, unit variance entries, the SCM can be written in a matrix form as 
\begin{equation}
\hat{\mathbf{R}}\left(n\right)=\mathbf{R}^{\frac{1}{2}}\mathbf{X}\mathbf{\Lambda}\left(n\right)\mathbf{X}^{*}\mathbf{R}^{\frac{1}{2}},
\end{equation}
where $\mathbf{X}$ is $m \times n$ matrix whose $k$th column is $\mathbf{x}\left(k\right)$ and $\mathbf{\Lambda}\left(n\right)=\left(1-\lambda\right)\diag\left(\lambda^{n-1}, \lambda^{n-2}, \cdots, 1\right)$. 
In the following, we prove that we can derive a closed-form expression for the loss function defined in (\ref{loss}) . Let $\mathbf{S}_n=\mathbf{X}\mathbf{\Lambda}\left(n\right)\mathbf{X}^{*}$. Then, the loss can be expressed as
\begin{equation}
\label{loss_close}
\begin{split}
Loss\left(n\right)&=\mathbb{E}\left \|\mathbf{S}_n^{-1}-\mathbf{I}_m\right\|^2_F \\
&=\mathbb{E}\textsf{Tr} \left[\mathbf{S}_n^{-1}-\mathbf{I}_m\right]^* \left[\mathbf{S}_n^{-1}-\mathbf{I}_m\right] \\
&=\mathbb{E}\textsf{Tr}\Biggl[\mathbf{S}_n^{-2}-2\mathbf{S}_n^{-1}+\mathbf{I}_m\Biggr] \\
&=m+\mathbb{E}\textsf{Tr}\left[\mathbf{S}_n^{-2}\right]  -2\mathbb{E} \textsf{Tr}\left[\mathbf{S}_n^{-1}\right]\\
\normalsize
&=m+m\mu_{\mathbf{\Lambda}(n)}\left(-2\right)-2m\mu_{\mathbf{\Lambda}(n)}\left(-1\right) \\
& = m\left(1+\mu_{\mathbf{\Lambda}(n)}\left(-2\right)-2\mu_{\mathbf{\Lambda}(n)}\left(-1\right)\right).
\end{split}
\end{equation} 

One interesting problem is to find the optimal $\lambda \in \left(0,1\right)$ denoted by $\lambda^*$ that minimizes the loss function. This can be performed by   evaluating the loss function with respect to $\lambda \in \left(0,1\right)$ and then picking the $\lambda$ that gives the lowest loss. To evaluate the loss function, one can resort to MonteCarlo simulations. This would involve high complexity since they should be repeated for each value of $\lambda$. The use of the provided closed-form expression represent thus a valuable alternative being at the same time accurate and easier to implement. 
For $m=3$ and $n=10$, we plot the estimation loss as a function of $\lambda$ using the theoretical expression derived in (\ref{loss_close}). As shown in Figure l, for all cases a minimum exist and thus the performance can be optimized accordingly.
\begin{figure}[t!]
        \centering
    \includegraphics[width=3.5in]{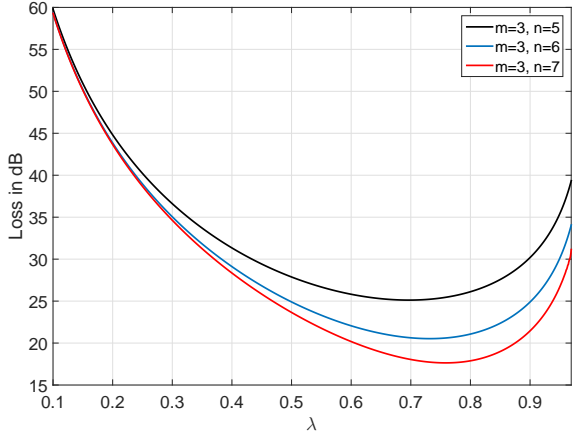}
\caption{The estimation loss as a function of $\lambda$: Theoretical expression in (\ref{loss_close}).}
\label{fig7}
\end{figure} 
\section{Conclusion} In this paper, we derived closed form expressions for the inverse order moments of general Gram matrices with one side correlation. Based on this formula, the exact average estimation error of the BLUE estimator has been derived and an accurate approximation for the LMMSE average estimation error was proposed in both high and low SNR regimes. Additionally, we have shown that our results can be used to evaluate the accuracy of  covariance matrix estimates.  
\section*{Appendix A}
\section*{Proof of Proposition 1}
As stated in section II, $\mathcal{M}_2\left(s\right)$ is given by:

\begin{align*}
&\mathcal{M}_2\left(s\right)=L\sum
_{j=r+1}^{m}\sum_{i=1}^{m}\mathcal{D}\left(i,j\right)\Gamma\left(s+j-1\right)\Biggl( \theta_{n-m+i}^{n-m+s+j-2} \\
& -\sum_{l=1}^{n-m}\sum_{k=1}^{n-m}\left[\mathbf{\Psi}^{-1}\right]_{k,l}\theta_{l}^{n-m+s+j-2}\theta_{n-m+i}^{k-1}\Biggr).
\end{align*}
The handling of $\mathcal{M}_2\left(s\right)$ does not pose any difficulty, the Gamma function being applied to  non-negative arguments. Interestingly, we can show that this term turns out to be equal to zero as $s$ goes to zero. 
To this end, notice that:
\small
\begin{align*}
&\lim _{s \rightarrow 0}\mathcal{M}_2\left(s-r+1\right)\\
&=L\sum
_{j=r+1}^{m}\sum_{i=1}^{m}\mathcal{D}\left(i,j\right)\Gamma\left(-r+j\right)\\& \times \Biggl( \theta_{n-m+i}^{n-m-r+j-1} -\sum_{l=1}^{n-m}\sum_{k=1}^{n-m}\left[\mathbf{\Psi}^{-1}\right]_{k,l}\theta_{l}^{n-m-r+j-1}\theta_{n-m+i}^{k-1}\Biggr) \\
&=L\sum_{j=r+1}^{m}\sum_{i=1}^m \left[\mathbf{\mathcal{D}}\right]_{i,j}\left[\mathbf{\mathcal{C}}\right]_{i,j-r} \\
&=L\sum_{j=r+1}^m \left[\mathbf{\mathcal{D}}^t\mathbf{\mathcal{C}}\right]_{j,j-r},
\end{align*}
\normalsize
where $\mathbf{\mathcal{D}}$ and $\mathbf{\mathcal{C}}$ are as defined in Lemma \ref{lemma:key}. 
Since $\mathbf{\mathcal{D}}$ is the cofactor of $\mathbf{\mathcal{C}}$, then $\mathbf{\mathcal{D}}^t\mathbf{\mathcal{C}}=\det\left(\mathbf{\mathcal{C}}\right) \mathbf{I}_{m}$. Therefore, $\left[\mathbf{\mathcal{D}}^t\mathbf{\mathcal{C}}\right]_{j,j-r}=0$ for $j=r+1, \cdots, m.$
\section*{Appendix B}
\section*{Proof of Proposition 2}
The handling of  $\mathcal{M}_1\left(s-r+1\right)$ is delicate because it involves evaluation of the Gamma function at negative integers. Hopefully, a compensation effect occurs due to the multiplicative term in front of the Gamma function. 
The proof relies on  a divide and conquer strategy that consists in decomposing $\mathcal{M}_1\left(s-r+1\right)$ into a sum of terms and then evaluating each term separately. To this end, we need to introduce the following notations.
\begin{equation}
\begin{split}
&\mathbf{\Psi}_s \triangleq \begin{bmatrix}
\theta_1^s &\theta_1^{1+s}  &\cdots   &\theta_1^{n-m+s-1} \\
\vdots  &\vdots   &\ddots   &\vdots  \\
 \theta_{n-m}^s&\theta_{n-m}^{1+s}  & \cdots  & \theta_{n-m}^{n-m+s-1}
\end{bmatrix} \\
& \mathbf{a}_{s,j}\triangleq\Biggl[\theta_1^{n-m+s-r+j-1}, \theta_2^{n-m+s-r+j-1}, \\&
 \cdots, \theta_{n-m}^{n-m+s-r+j-1} \Biggr]^t \\
&\mathbf{b}_{s,i}\triangleq\left[\theta_{n-m+i}^s,\theta_{n-m+i}^{1+s}, \cdots, \theta_{n-m+i}^{n-m+s-1}\right]^t \\
&\mathbf{b}_i\triangleq \left[1,\theta_{n-m+i}, \cdots, \theta_{n-m+i}^{n-m-1}\right]^t\\
&\mathbf{e}_{k}  \triangleq\left[0, \cdots,0,1, \mathbf{zeros}\left(k\right)\right]^t, \: k=0, \cdots, n-m-1.\\
\end{split}
\end{equation}
Using the previously defined varaibles, we can rewrite $\mathcal{M}_1\left(s-r+1\right)$ as in (\ref{M1}) (on top of the next page).
\begin{figure*}[!t]
 \begin{equation}
 \label{M1}
 \begin{split}
 \mathcal{M}_1\left(s-r+1\right)&=
L\sum
_{j=1}^{r}\sum_{i=1}^{m}\mathcal{D}\left(i,j\right)\Gamma\left(s-r+j\right) \Biggl( \theta_{n-m+i}^{n-m+s-r+j-1}  -\sum_{l=1}^{n-m}\sum_{k=1}^{n-m}\left[\mathbf{\Psi}^{-1}\right]_{k,l}\theta_{l}^{n-m+s-r+j-1}\theta_{n-m+i}^{k-1}\Biggr) \\
&=L\sum
_{j=1}^{r}\sum_{i=1}^{m}\mathcal{D}\left(i,j\right)\Gamma\left(s-r+j\right) \Biggl(\theta_{n-m+i}^{n-m+s-r+j-1}-\mathbf{b}_i^t \mathbf{\Psi}^{-1}\mathbf{a}_{s,j}\Biggr) \\ 
&=L\sum
_{j=1}^{r}\sum_{i=1}^{m}\mathcal{D}\left(i,j\right)\Gamma\left(s-r+j\right)\Biggl(\theta_{n-m+i}^{n-m+s-r+j-1}
-\mathbf{b}_{i}^t \mathbf{\Psi}_s^{-1}\mathbf{a}_{s,j}\Biggr)+L\sum
_{j=1}^{r}\sum_{i=1}^{m}\mathcal{D}\left(i,j\right)\Gamma\left(s-r+j\right) \\& \times
\mathbf{b}_i^t\left(\mathbf{\Psi}_s^{-1}-\mathbf{\Psi}^{-1}\right)\mathbf{a}_{s,j} \\
&=L\sum
_{j=1}^{r}\sum_{i=1}^{m}\mathcal{D}\left(i,j\right)\Gamma\left(s-r+j\right)\Biggl(\theta_{n-m+i}^{n-m+s-r+j-1} -\mathbf{b}_{s,i}^t \mathbf{\Psi}_s^{-1}\mathbf{a}_{s,j}\Biggr)+L\sum
_{j=1}^{r}\sum_{i=1}^{m}\mathcal{D}\left(i,j\right)\Gamma\left(s-r+j\right) \\
&\times \left(\mathbf{b}_{s,i}^t-\mathbf{b}_{i}^t\right)\mathbf{\Psi}_s^{-1}\mathbf{a}_{s,j}+L\sum
_{j=1}^{r}\sum_{i=1}^{m}\mathcal{D}\left(i,j\right)\Gamma\left(s-r+j\right) 
\mathbf{b}_i^t\left(\mathbf{\Psi}_s^{-1}-\mathbf{\Psi}^{-1}\right)\mathbf{a}_{s,j}
\end{split} 
 \end{equation}
\hrulefill
\vspace*{4pt}
\end{figure*}
 
 The first term in equation (\ref{M1}) is equal to zero. This can be seen by noticing that $\mathbf{\Psi}_s \mathbf{e}_{r-j}=\mathbf{a}_{s,j}$ and $\mathbf{b}_{s,i}^t\mathbf{e}_{r-j}=\theta_{n-m+i}^{n-m+s-r+j-1}$. Thus, $\mathbf{\Psi}_s^{-1}\mathbf{a}_{s,j}=\mathbf{e}_{r-j}$ and consequently $\mathbf{b}_{s,i}^t\mathbf{\Psi}_s^{-1}\mathbf{a}_{s,j}=\theta_{n-m+i}^{n-m+s-r+j-1} $. \\ It remains thus  to deal with the last two terms. Using a  Taylor approximation of $\mathbf{b}_{s,i}$ as $s$ approaching $0$, we have 
 \begin{equation}
 \begin{split}
 \mathbf{b}_{s,i}-\mathbf{b}_{i} 
 &=s\Biggl[\log \left(\theta_{n-m+i}\right),\theta_{n-m+i}\log \left(\theta_{n-m+i}\right),\\
 & \cdots,\theta_{n-m+i}^{n-m-1} \log \left(\theta_{n-m+i}\right) \Biggr]^t
 +o\left(s\right) \\
 &=s\log\left(\theta_{n-m+i}\right)\mathbf{b}_i+o\left(s\right).
 \end{split}
 \end{equation}
 To deal with the Gamma function evaluated at non positive integers, we rely on the result of the following lemma.
 \begin{lemma}{\cite{gamma}}
 For non positive arguments $-k$, $k=0, 1, 2, \cdots$, the Gamma function can be evaluated as 
  \begin{equation}
  \lim _{s \rightarrow 0} \frac{\Gamma\left(s-k\right)}{\Gamma\left(s\right)}=\frac{\left(-1\right)^k}{k!},
  \end{equation}
  where $\Gamma\left(s\right) = \frac{1}{s} +o(s)$ as $s$ approaches $0$.
 \end{lemma}
  Thus, $\Gamma\left(s-r+j\right) \underset{s\to 0}{=} \frac{\left(-1\right)^{r-j}}{s\left(r-j\right)!} +o(s)$. Therefore , as $s$ approaches 0, we have 
 \begin{align*}
 \Gamma\left(s-r+j\right) \left(\mathbf{b}_{s,i}^t-\mathbf{b}_{i}^t\right)\mathbf{\Psi}_s^{-1}\mathbf{a}_{s,j}\\ = \frac{\left(-1\right)^{r-j}\log\left(\theta_{n-m+i}\right)}{\left(r-j\right)!} \mathbf{b}_i^t\mathbf{\Psi}^{-1}\mathbf{a}_{j} +o(s),
 \end{align*}
 Finally, to deal with the last term, we use the following resolvent identity : 
 \begin{align*}
 \mathbf{\Psi}_s^{-1}-\mathbf{\Psi}^{-1}=\mathbf{\Psi}_s^{-1}\left(\mathbf{\Psi}-\mathbf{\Psi}_s\right)\mathbf{\Psi}^{-1}
\end{align*}  
We also make use of the fact  that as $s$ approaches $0$:
\begin{align*}
\left(\mathbf{\Psi}-\mathbf{\Psi}_s\right) \underset{s\to 0}{=}-s \tilde{\mathbf{\Psi}} +o(s)
\end{align*} 
where $$\tilde{\mathbf{\Psi}}=\boldsymbol{\Phi}\boldsymbol{\Psi}$$
with $\boldsymbol{\Phi}=\textsf{diag}\left(\log\left(\theta_1\right), \log\left(\theta_2\right),\cdots, \log\left(\theta_{n-m} \right)\right)$.
Thus, as $s$ approaches $0$, we have
\small
\begin{align*}
\Gamma\left(s-r+j\right) \mathbf{b}_i^t\left(\mathbf{\Psi}_s^{-1}-\mathbf{\Psi}^{-1}\right)\mathbf{a}_{s,j} \\
 \underset{s\to 0}{=} \frac{\left(-1\right)^{r+1-j}}{\left(r-j\right)!}\mathbf{b}_i^t\mathbf{\Psi}^{-1}\tilde{\mathbf{\Psi}}\mathbf{\Psi}^{-1}\mathbf{a}_j+o(s).
\end{align*}
\normalsize 
Finally, we have the following limit 
\begin{align*}
&\lim _{s \rightarrow 0} \mathcal{M}_1\left(s-r+1\right)\\&=
L\sum
_{j=1}^{r}\sum_{i=1}^{m}\mathcal{D}\left(i,j\right)\Biggl[\frac{\left(-1\right)^{r-j}\log\left(\theta_{n-m+i}\right)}{\left(r-j\right)!} \mathbf{b}_i^t\mathbf{\Psi}^{-1}\mathbf{a}_{j} \\
+&\frac{\left(-1\right)^{r+1-j}}{\left(r-j\right)!}\mathbf{b}_i^t\mathbf{\Psi}^{-1}\tilde{\mathbf{\Psi}}\mathbf{\Psi}^{-1}\mathbf{a}_j\Biggr]. \\
\end{align*}
This expression can be further simplified by noticing that $\tilde{\mathbf{\Psi}}\mathbf{\Psi}^{-1}=\boldsymbol{\Phi}$. Finally, we have 
\begin{align*}
&\lim _{s \rightarrow 0} \mathcal{M}_1\left(s-r+1\right)\\
&=L\sum
_{j=1}^{r}\sum_{i=1}^{m}\mathcal{D}\left(i,j\right)\frac{\left(-1\right)^{r-j}}{\left(r-j\right)!}\mathbf{b}_i^t\mathbf{\Psi}^{-1}\Biggl[\log\left(\theta_{n-m+i}\right) \mathbf{I}_{n-m} \\
&-\boldsymbol{\Phi} \Biggr]\mathbf{a}_j \\
& = L\sum
_{j=1}^{r}\sum_{i=1}^{m}\mathcal{D}\left(i,j\right)\frac{\left(-1\right)^{r-j}}{\left(r-j\right)!}\mathbf{b}_i^t\mathbf{\Psi}^{-1}\mathbf{D}_i\mathbf{a}_j,
\end{align*}
thereby ending up the proof of the proposition. 
\section*{Appendix C}
\section*{Proof of Theorem 2}
We start the proof by noticing that $m\left(z\right)$ satisfies:
\begin{align*}
-z\underline{m}\left(z\right)+\frac{n}{m}-\frac{1}{m}\sum_{k=1}^n\frac{1}{1+\left[\mathbf{D}\right]_{k,k}\underline{m}\left(z\right)}=1.
\end{align*}
Let $f_k\left(z\right)=-\frac{1}{1+\left[\mathbf{D}\right]_{k,k}\underline{m}\left(z\right)}$. Then, the above equation becomes:
\begin{align*}
-z\underline{m}\left(z\right)+\frac{n}{m}+\frac{1}{m}\sum_{k=1}^nf_k\left(z\right)=1.
\end{align*}
Taking the $p-1$ derivative of the above equation and set $z=0$, we have:
\begin{align}
\label{recursivem}
p\underline{m}^{\left(p-1\right)}\left(0\right)=\frac{1}{m}\sum_{k=1}^nf^{\left(p\right)}_k\left(0\right).
\end{align}
On the other hand, functions $f_k\left(z\right)$ satisfy:
\begin{align} \label{appCeq4}
-f_k\left(z\right)-\left[\mathbf{D}\right]_{k,k}f_k\left(z\right)\underline{m}\left(z\right)=1.
\end{align}
Taking the $p$-th derivative of equation (\ref{appCeq4}), we get:
\begin{align*}
-f^{\left(p\right)}_k-\left[\mathbf{D}\right]_{k,k}\sum_{l=0}^p  \binom {p} {l}f^{\left(p-l\right)}_k\underline{m}^{\left(l\right)}.
\end{align*}
or equivalently,
\begin{equation}
\label{recursivefk}
f^{\left(p\right)}_k+\sum_{l=1}^{p}\binom {p} {l}\frac{\left[\mathbf{D}\right]_{k,k}\underline{m}^{\left(l\right)}f^{\left(p-l\right)}_k}{1+\left[\mathbf{D}\right]_{k,k}\underline{m}\left(0\right)}=0
\end{equation}
Hence, by separating the first term of the sum in (\ref{recursivefk}), we obtain
\begin{align*}
f^{\left(p\right)}_k+\frac{\left[\mathbf{D}\right]_{k,k}\underline{m}^{\left(p\right)}f^{\left(0\right)}_k}{1+\left[\mathbf{D}\right]_{k,k}\underline{m}\left(0\right)}+\sum_{l=1}^{p-1}\binom {p} {l}\frac{\left[\mathbf{D}\right]_{k,k}\underline{m}^{\left(l\right)}f^{\left(p-l\right)}_k}{1+\left[\mathbf{D}\right]_{k,k}\underline{m}\left(0\right)}=0.
\end{align*}
Combining (\ref{recursivem}) and (\ref{recursivefk}) and taking the sum over $k$ of the above equation, we get:
\begin{equation}
\label{ierative}
\begin{split}
 p\underline{m}^{\left(p-1\right)} &+\frac{\underline{m}^{\left(p\right)}}{m}\sum_
{k=1}^n \frac{\left[\mathbf{D}\right]_{k,k}f^{\left(0\right)}_k}{1+\left[\mathbf{D}\right]_{k,k}\underline{m}\left(0\right)} \\
&+\frac{1}{m}\sum_{k=1}^n\sum_{l=1}^{p-1}\binom {p} {l}\frac{\left[\mathbf{D}\right]_{k,k}\underline{m}^{\left(l\right)}f^{\left(p-l\right)}_k}{1+\left[\mathbf{D}\right]_{k,k}\underline{m}\left(0\right)}=0,
\end{split}
\end{equation}
where $\underline{m}\left(0\right)$ is the unique solution to the fixed point equation:
\begin{equation}
\label{m0}
\underline{m}\left(0\right)=\frac{1}{\frac{1}{m}\sum_{k=1}^n\frac{\left[\mathbf{D}\right]_{k,k}}{1+\left[\mathbf{D}\right]_{k,k}\underline{m}\left(0\right)}}.
\end{equation}
This ends up the proof of the Theorem.
\section*{Appendix D}
\section*{Proof of Theorem 4}
The proof of this theorem is based on a Taylor approximation of the LMMSE average error. 
\begin{enumerate}
\item High SNR regime ($\sigma_x^2 \gg 1 $): \\
 As stated in equation (\ref{covarianceMMSE})
 \begin{align*}
 \mathbf{\Sigma}_{e,lmmse}=\left(\frac{1}{\sigma_x^2}\mathbf{I}_n+\mathbf{H}^*\mathbf{\Sigma}_z^{-1}\mathbf{H}\right)^{-1}.
\end{align*} 
By setting $\mathbf{\Sigma}_z^{-1}=\mathbf{\Lambda}$, we have 
\begin{align*}
\mathbf{\Sigma}_{e,lmmse}=\left(\frac{1}{\sigma_x^2}\mathbf{I}_n+\mathbf{S}\right)^{-1}.
\end{align*}
Since $\frac{1}{\sigma_x^2} \ll 1$, then by Tayor expansion around $\mathbf{0}$, the trace of $\mathbf{\Sigma}_{e,lmmse}$ can be expressed as:
\begin{align}
\textsf{Tr}\left(\mathbf{\Sigma}_{e,lmmse}\right)&=\sum_{k=0}^{\infty}\frac{\left(-1\right)^k}{\sigma_x^{2k}}\textsf{Tr}\left(\mathbf{S}^{-k-1}\right),\\
&=\sum_{k=0}^{l}\frac{\left(-1\right)^k}{\sigma_x^{2k}}\textsf{Tr}\left(\mathbf{S}^{-k-1}\right) +o\left(\frac{1}{\sigma_x^{2l}}\right)
\end{align}
where $l\leq p-1$ is the order at which we truncate the Taylor expansion. Note that we impose the condition $r\leq p-1$ since the moments $\mu_{\mathbf{\Lambda}}\left(-k-1\right)$ are only defined for $1 \leq k \leq p-1$.
 Upon applying the expectation, we get:
\begin{equation}
\begin{split}
\mathbb{E}_{\mathbf{H}}\{ \Vert\hat{\mathbf{x}}_{lmmse}-\mathbf{x}\Vert^2\}&=\sum_{k=0}^{l} \frac{\left(-1\right)^k}{\sigma_x^{2k}} \mu_{\mathbf{\Lambda}}\left(-k-1\right)\\ &+o\left(\sigma_x^{-2l}\right)
\end{split}
\end{equation}

 \item Low SNR regime ($\sigma_x^2 \ll 1$): \\
We proceed similarly as the high SNR regime. For that, we make some manipulations on $\mathbf{\Sigma}_{e,lmmse}$ as follows:
\begin{align*}
\mathbf{\Sigma}_{e,lmmse} &=\left(\frac{1}{\sigma_x^2}\mathbf{I}_n+\mathbf{S}\right)^{-1} \\
&=\sigma_x^2 \left(\mathbf{I}_n+\sigma_x^2 \mathbf{S}\right)^{-1} \\
&=\sigma_x^2 \left(\sigma_x^2\mathbf{I}_n+\mathbf{S}^{-1}\right)^{-1}\mathbf{S}^{-1}
\end{align*}
Using the same approach for proving Theorem \ref{approximation}. 1), we get 
\begin{equation}
\mathbb{E}_{\mathbf{H}}\{ \Vert\hat{\mathbf{x}}_{lmmse}-\mathbf{x}\Vert^2\}=m\sum_{k=0}^{\infty} \left(-1\right)^k\sigma_x^{2k+2} \mu_{\mathbf{\Lambda}}\left(k\right)
\end{equation}
As seen in the previous equation, we retain all the positive moments since they exist and can be computed by means of Theorem 1. \\
This completes the proof of Theorem \ref{approximation}
\end{enumerate}
\bibliographystyle{IEEEtran}
\bibliography{References}

\begin{thebibliography}{10}
\providecommand{\url}[1]{#1}
\csname url@samestyle\endcsname
\providecommand{\newblock}{\relax}
\providecommand{\bibinfo}[2]{#2}
\providecommand{\BIBentrySTDinterwordspacing}{\spaceskip=0pt\relax}
\providecommand{\BIBentryALTinterwordstretchfactor}{4}
\providecommand{\BIBentryALTinterwordspacing}{\spaceskip=\fontdimen2\font plus
\BIBentryALTinterwordstretchfactor\fontdimen3\font minus
  \fontdimen4\font\relax}
\providecommand{\BIBforeignlanguage}[2]{{%
\expandafter\ifx\csname l@#1\endcsname\relax
\typeout{** WARNING: IEEEtran.bst: No hyphenation pattern has been}%
\typeout{** loaded for the language `#1'. Using the pattern for}%
\typeout{** the default language instead.}%
\else
\language=\csname l@#1\endcsname
\fi
#2}}
\providecommand{\BIBdecl}{\relax}
\BIBdecl

\bibitem{couillet-08}
\emph{{Free Deconvolution for OFDM Multicell SNR Detection}}, 2008.

\bibitem{kammoun_yao}
J.~Yao, A.~Kammoun, and J.~Najim, ``Eigenvalue estimation of parameterized
  covariance matrices of large dimensional data,'' \emph{IEEE Transactions on
  Signal Processing}, vol.~60, no.~11, pp. 5893 --5905, nov. 2012.

\bibitem{Ryan}
O.~Ryan and M.~Debbah, ``Asymptotic behavior of random vandermonde matrices
  with entries on the unit circle,'' \emph{IEEE Transactions on Information
  Theory,}, vol.~55, no.~7, pp. 3115--3147, July 2009.

\bibitem{jacob}
J.~Hoydis, M.~Debbah, and M.~Kobayashi, ``Asymptotic moments for interference
  mitigation in correlated fading channels,'' in \emph{Information Theory
  Proceedings (ISIT)}, July 2011, pp. 2796--2800.

\bibitem{maiwald}
{D. Maiwald and D. Kraus}, ``{Calculation of moments of complex Wishart and
  Complex Inverse Wishart Distributed Matrices},'' in \emph{IEEE Proceeding
  Radar Sonar and Navigation}, 2000.

\bibitem{letac-04}
{G. Letac and H. Massam}, ``{All Invariant Moments of the Wishart
  Distribution},'' \emph{Scandinavian Journal of Statistics}, vol.~31, no.~2,
  pp. 295--318, Jun. 2004.

\bibitem{giusi}
G.~Alfano, A.~M. Tulino, A.~{L}ozano, and S.~{V}erdu, ``{C}apacity of {MIMO}
  {C}hannels with {O}ne-sided {C}orrelation,'' \emph{ISSSTA}, August 2004.

\bibitem{edelman-phd}
A.~Edelman, ``Eigenvalues and {C}ondition {N}umbers of {R}andom {M}atrices,''
  Ph.D. dissertation, Massachusets Institute of Technology, 1989.

\bibitem{silverstein}
J.~W. {S}ilverstein and Z.~D. {B}ai, ``{O}n the empirical distribution of
  {E}igenvalues of a {C}lass of {L}arge {D}imensional {R}andom {M}atrices,''
  \emph{Journal of Multivariate Analysis}, vol.~54, pp. 175--192, May 2002.

\bibitem{kailath}
T.~Kailath, A.~H. {S}ayed, and B.~{H}assibi, \emph{{L}inear
  {E}stimation}.\hskip 1em plus 0.5em minus 0.4em\relax Prentice Hall, 2000.

\bibitem{sayed}
A.~Sayed, \emph{Fundamentals of Adaptive Filtering}.\hskip 1em plus 0.5em minus
  0.4em\relax John Wiley \& Sons, 2003.

\bibitem{vpoor}
H.~V. {P}oor, \emph{{A}n {I}ntroduction to {S}ignal {D}etection and
  {E}stimation}.\hskip 1em plus 0.5em minus 0.4em\relax Springer-Verlag, 1988.

\bibitem{MITthesis}
{M}ilutin {P}ajovic, ``{T}he {D}evelopment and {A}pplication of {R}andom
  {M}atrix {T}heory in {A}daptive {S}ignal {P}rocessing in the {S}ample
  {D}eficient {R}egime,'' Ph.D. dissertation, Massachusetts Institute of
  Technology, 2014.

\bibitem{estimationbook}
H.~L. {V}an {T}rees., \emph{{O}ptimum {A}rray {P}rocessing: {D}etection,
  {E}stimation and {M}odulation {T}heory}.\hskip 1em plus 0.5em minus
  0.4em\relax John Wiley and Sons, 2002.

\bibitem{couillet}
R.~{C}ouillet and M.~{D}ebbah, \emph{{R}andom {M}atrix {M}ethods for {W}ireless
  {C}ommunications}.\hskip 1em plus 0.5em minus 0.4em\relax Cambridge
  University Press, 2011.

\bibitem{yang-94}
R.~Yang and J.~O. Berger, ``{Estimation of a covariance matrix using the
  Reference Prior},'' \emph{Annals of Statistics}, vol.~22, no.~3, pp.
  1195--1211, 1994.

\bibitem{gamma}
M.~A. Chaudhry and S.~M. Zubair, \emph{On a {C}lass of {I}ncomplete {G}amma
  {F}unction with {A}pplications}.\hskip 1em plus 0.5em minus 0.4em\relax Boca
  Raton-London-Ney York-Washington, D.C.: Chapman \& Hall/CRC, 2002.

\end{thebibliography}
\end{document}